\documentclass[oribibl, 11pt]{llncs}
\usepackage{fullpage}

\usepackage{amsmath}
\usepackage{amssymb}
\usepackage{graphicx}

\usepackage{algorithm}
\usepackage[noend]{algpseudocode}

\usepackage{float}

\usepackage{comment}

\usepackage{subcaption} 
\usepackage{caption}

\usepackage[pdfstartview={XYZ null null 1.00}, pdfpagemode=None,colorlinks=true,citecolor=blue,linkcolor=blue]{hyperref}

\usepackage[
textsize=scriptsize]{todonotes}

\usepackage{color}
\usepackage{hyperref}
\usepackage{gensymb}

\usepackage{fancyhdr}
\makeatletter
\def\blfootnote{\xdef\@thefnmark{}\@footnotetext}
\def\BState{\State\hskip-\ALG@thistlm}

\makeatother

\newcommand{\removed}[1]{}

\newcommand{\rem}{\mathbf}

\newcommand\tab[1][1cm]{\hspace*{#1}}

\newcommand\given[1][]{\:#1\vert\:}

\definecolor{DarkRed}{RGB}{182,11,1}

\usepackage{mathtools}

\DeclarePairedDelimiter\ceil{\lceil}{\rceil}

\algnewcommand{\IfThenElse}[3]{
	\State \algorithmicif\ #1\ \algorithmicthen\ #2\ \algorithmicelse\ #3}

\title{Fast Approximate Counting and Leader Election in Populations\thanks{All authors were supported by the EEE/CS initiative NeST. The last author was also supported by the Leverhulme Research Centre for Functional Materials Design.}}

\author{Othon Michail\inst{1} \and Paul G. Spirakis\inst{1,2} \and Michail Theofilatos\inst{1}}
\institute{Department of Computer Science, University of Liverpool, UK \and Computer Engineering and Informatics Department, University of Patras, Greece\\
	Email:\email{ \{Othon.Michail, P.Spirakis, Michail.Theofilatos\}@liverpool.ac.uk}}

\begin{document}

\maketitle

\begin{abstract}
We study the problems of leader election and population size counting for \emph{population protocols}:
networks of finite-state anonymous agents that interact randomly under a uniform random scheduler.
We show a protocol for leader election that terminates in $O(\log_m(n) \cdot \log_2 n)$ parallel time,
where $m$ is a parameter,
using $O(\max\{m,\log n\})$ states.
By adjusting the parameter $m$ between a constant and $n$,
we obtain a single leader election protocol whose time and space can be smoothly traded off between $O(\log^2 n)$ to $O(\log n)$ time and $O(\log n)$ to $O(n)$ states.
Finally, we give a protocol which provides an upper bound $\hat{n}$ of the size $n$ of the population, where $\hat{n}$ is at most $n^a$ for some $a>1$. This protocol assumes the existence of a unique leader in the population and stabilizes in $\Theta{(\log{n})}$ parallel time, using constant number of states in every node, except the unique leader which is required to use $\Theta{(\log^2{n})}$ states.

\end{abstract}

\noindent
\textbf{Keywords:} population protocol, epidemic, leader election, counting, approximate counting, polylogarithmic time protocol\newline

\section{Introduction}
\label{sec:intro}

\emph{Population protocols} \cite{AADFP06} are networks that consist of very weak computational entities (also called \textit{nodes} or \textit{agents}), regarding their individual capabilities. These networks have been shown that are able to construct complex shapes \cite{MS16a} and perform complex computational tasks when they work collectively.
Leader Election, which is a fundamental problem in distributed computing, is the process of designating a single agent as the coordinator of some task distributed among several nodes. The nodes communicate among themselves in order to decide which of them will get into the \textit{leader} state.
\textit{Counting} is also a fundamental problem in distributed computing, where nodes must determine the size $n$ of the population. Finally, we call \textit{Approximate Counting} the problem in which nodes must determine an estimation $k$ of the population size $n$. Counting can be then considered as a special case of population size estimation, where $k=n$.

Many distributed tasks require the existence of a leader prior to the execution of the protocol and, furthermore, some knowledge about the system (for instance the size of the population) can also help to solve these tasks more efficiently with respect both to time and space.

Consider the setting in which an agent is in an initial state \textit{a}, the rest $n-1$ agents are in state \textit{b} and the only existing transition is $(a,b) \rightarrow (a,a)$. This is the \textit{one-way epidemic} process and it can be shown that the expected time to convergence under the uniform random scheduler is $\Theta(n\log{n})$ (e.g., \cite{AAE08}), thus $\Theta(\log{n})$ \textit{parallel time}.
In this work, we make an extensive use of epidemics, which means that information is being spread throughout the population, thus all nodes will obtain this information in $O(\log{n})$ expected parallel time. We use this property to construct an algorithm that solves the \textit{Leader Election} problem. In addition, by observing the rate of the epidemic spreading under the uniform random scheduler, we can extract valuable information about the population. This is the key idea of our \textit{Approximate Counting} algorithm.

\subsection{Related Work}
\label{subsec:related}

The framework of population protocols was first introduced by Angluin et al. \cite{AADFP06} in order to model the interactions in networks between small resource-limited mobile agents. When operating under a uniform random scheduler, population protocols are formally equivalent to a restricted version of stochastic Chemical Reaction Networks (CRNs), which model chemistry in a well-mixed solution \cite{SCWB08}. ``CRNs are widely used to describe information processing occurring in natural cellular regulatory networks, and with upcoming advances in synthetic biology, CRNs are a promising programming language for the design of artificial molecular control circuitry'' \cite{CDS14, Do14}. Results in both population protocols and CRNs can be transfered to each other, owing to a formal equivalence between these models.

Angluin et al. \cite{AAER07} showed that all predicates stably computable in population protocols (and certain generalizations of it) are semilinear. Semilinearity persists up to $o(\log\log n)$ local space but not more than this \cite{CMNPS11}. Moreover, the computational power of population protocols can be increased to the commutative subclass of $\rem{NSPACE}(n^2)$, if we allow the processes to form connections between each other that can hold a state from a finite domain \cite{MCS11}, or by equipping them with unique identifiers, as in \cite{GR09}.
For introductory texts to population protocols the interested reader is encouraged to consult \cite{AR09,MCS11} and \cite{MS18} (the latter discusses population protocols and related developments as part of a more general overview of the emerging theory of dynamic networks).

Optimal algorithms, regarding the time complexity of fundamental tasks in distributed networks, for example leader election and majority, is the key for many distributed problems. For instance, the help of a central coordinator can lead to simpler and more efficient protocols \cite{AAE08}. There are many solutions to the problem of leader election, such as in networks with nodes having distinct labels or anonymous networks \cite{An80, ASW85, AG15, GS18, FJ06}.

Although the availability of an initial leader does not increase the computational power of standard population protocols (in contrast, it does in some settings where faults can occur \cite{DFI17}), still it may allow faster computation. Specifically, the fastest known population protocols for semilinear predicates without a leader take as long as linear parallel time to converge ($\Theta(n)$). On the other hand, when the process is coordinated by a unique leader, it is known that any semilinear predicate can be stably computed with polylogarithmic expected convergence time ($O(\log^5 n)$) \cite{AAE06}.

For several years, the best known algorithm for leader election in population protocols was the pairwise-elimination protocol of Angluin et al. \cite{AADFP06}, in which all nodes are leaders in state $l$ initially and the only effective transition is $(l,l)\rightarrow (l,f)$. This protocol always stabilizes to a configuration with unique leader, but this takes on average linear time. Recently, Doty and Soloveichik \cite{DS15} proved that not only this, but any standard population protocol requires linear time to solve leader election. This immediately led the research community to look into ways of strengthening the population protocol model in order to enable the development of sub-linear time protocols for leader election and other problems (note that Belleville, Doty, and Soloveichik \cite{BDS17} recently showed that such linear time lower bounds hold for a larger family of problems and not just for leader election). Fortunately, in the same way that increasing the local space of agents led to a substantial increase of the class of computable predicates \cite{CMNPS11}, it has started to become evident that it can also be exploited to substantially speed-up computations. Alistarh and Gelashvili \cite{AG15} proposed the first sub-linear leader election protocol, which stabilizes in $O(\log^3n)$ parallel time, assuming $O(\log^3n)$ states at each agent. In another recent work, Gasieniec and Stachowiak \cite{GS18} designed a space optimal ($O(\log\log{n})$ states) leader election protocol, which stabilises in $O(\log^2n)$ parallel time. They use the concept of phase clocks (introduced in \cite{AAE08} for population protocols), which is a synchronization and coordination tool in distributed computing. General characterizations, including upper and lower bounds, of the trade-offs between time and space in population protocols were recently achieved in \cite{AAEGR17}. Moreover, some papers \cite{MOKY12,DDFSV17} have studied leader election in the mediated population protocol model.

For counting, the most studied case is that of \emph{self-stabilization},
which makes the strong adversarial assumption that arbitrary corruption of memory is possible in any agent at any time,
and promises only that eventually it will stop.
Thus, the protocol must be designed to work from any possible configuration of the memory of each agent.
It can be shown that counting is \emph{impossible} without having one agent (the ``base station'')
that is protected from corruption~\cite{beauquier2007self}.
In this scenario
$\Theta(n \log n)$ time is sufficient~\cite{beauquier2015space} and necessary~\cite{AspnesBBS2016} for self-stabilizing counting.

In the less restrictive setting in which all nodes start from the same state (apart possibly from a unique leader and/or unique ids),
not much is known. In a recent work, Michail \cite{M15} proposed a terminating protocol in which a pre-elected leader equipped with two $n$-counters computes an approximate count between $n/2$ and $n$ in $O(n\log{n})$ parallel time with high probability. The idea is to have the leader implement two competing processes, running in parallel. The first process counts the number of nodes that have been encountered once, the second process counts the number of nodes that have been encountered twice, and the leader terminates when the second counter catches up the first. In the same paper, also a version assuming unique ids instead of a leader was given.

The task of counting has also been studied in the related context of worst-case dynamic networks \cite{IK14, KLO10, MCS13, LBBC14, CFQS12}.

\subsection{Contribution}
In this work we employ the use of simple epidemics in order to provide efficient solutions to approximate counting the size of a population of agents and also to leader election in populations. Our model is that of population protocols. Our goal for both problems is to get polylogarithmic parallel time and to also use small memory per agent. First, we show how to approximately count a population fast (with a leader) and then we show how to elect a leader (very fast) if we have a crude population estimate. \\
\textit{(a)} We start by providing a protocol which provides an upper bound $\hat{n}$ of the size $n$ of the population, where $\hat{n}$ is at most $n^a$ for some $a>1$. This protocol assumes the existence of a unique leader in the population. The runtime of the protocol until stabilization is $\Theta(\log{n})$ parallel time. Each node except the unique leader uses only a constant number of states. However, the leader is required to use $\Theta(\log^2{n})$ states. \\
\textit{(b)} We then look into the problem of electing a leader. We assume an approximate knowledge of the size of the population (i.e., an estimate $\hat{n}$ of at most $n^a$, where $n$ is the population size) and provide a protocol (parameterized by the size $m$ of a counter for drawing local random numbers) that elects a unique leader w.h.p. in $O(\frac{\log^2{n}}{\log{m}})$ parallel time, with number of states $O(\max\{m,\log{n}\})$ per node. 

\section{The model}
\label{sec:model}

In this work, the system consists of a population \textit{V} of \textit{n} distributed and anonymous (i.e., do not have unique IDs) \textit{processes}, also called \textit{nodes} or \textit{agents}, that are capable to perform local computations. Each of them is executing as a deterministic state machine from a finite set of states $Q$ according to a transition function $\delta: Q \times Q \rightarrow Q \times Q$. Their interaction is based on the probabilistic (uniform random) scheduler, which picks in every discrete step a random edge from the complete graph $G$ on $n$ vertices. When two agents interact, they mutually access their local states, updating them according to the transition function $\delta$. The transition function is a part of the population protocol which all nodes store and execute locally.

\textit{The time is measured as the number of steps until stabilization, divided by $n$ (parallel time)}. The protocols that we propose do not enable or disable connections between nodes, in contrast with \cite{MS16a}, where Michail and Spirakis considered a model where a (virtual or physical) connection between two processes can be in one of a finite number of possible states. The transition function that we present throughout this paper, follows the notation $(x,y) \rightarrow (z,w)$, which refers to the process states before \textit{(x and y)} and after \textit{(z and w)} the interaction, that is, the transition function maps pairs of states to pairs of states. \\

\paragraph{The Leader Election Problem.} The problem of leader election in distributed computing is for each node eventually to decide whether it is a leader or not subject to only one node decides that it is the leader. An algorithm $A$ solves the leader election problem if eventually the states of agents are divided into \textit{leader} and \textit{follower}, a leader remains elected and a follower can never become a leader. In every execution, exactly one agent becomes leader and the rest determine that they are not leaders. All agents start in the same initial state $q$ and the output is $O=\{leader, follower\}$. A randomized algorithm $R$ solves the leader election problem if eventually only one leader remains in the system w.h.p.

\paragraph{Approximate Counting Problem.}
We define as \textit{Approximate Counting} the problem in which a leader must determine an estimation $\hat{n}$ of the population size, where $\frac{\hat{n}}{a} < n < \hat{n}$. We call $a$ the estimation parameter.

\section{Fast Counting with a unique leader}
\label{sec:counting_with_unique_leader}

In this section we present our \textit{Approximate Counting} protocol. The protocol is presented in Section \ref{subsec:counting_with_unique_leader_abstract}. In Section \ref{subsec:counting_with_unique_leader_analysis} we prove the correctness of our protocol and finally, in Section \ref{sec:all_experiments}, experiments that support our analysis can be found.

\subsection{Abstract description and protocol}
\label{subsec:counting_with_unique_leader_abstract}

In this section, we construct a protocol which solves the problem of approximate counting.
Our probabilistic algorithm for solving the approximate counting problem requires a unique leader who is responsible to give an estimation on the number of nodes. It uses the epidemic spreading technique and it stabilizes in $O(\log{n})$ parallel time. There is initially a unique leader $l$ and all other nodes are in state $q$. The leader $l$ stores two counters in its local memory, initially both set to 0. We use the notation $l_{(c_q,c_a)}$, where $c_q$ is the value of the first counter and $c_a$ is the value of the second one. The leader, after the first interaction starts an epidemic by turning a $q$ node into an $a$ node. Whenever a $q$ node interacts with an $a$ node, its state becomes $a$ $((a, q) \rightarrow (a, a))$. The first counter $c_q$ is being used for counting the $q$ nodes and the second counter $c_a$ for the $a$ nodes, that is, whenever the leader $l$ interacts with a $q$ node, the value of the counter $c_q$ is increased by one and whenever $l$ interacts with an $a$ node, $c_a$ is increased by one. The termination condition is $c_q = c_a$ and then the leader holds a constant-factor approximation of $\log{n}$, which we prove that with high probability is $2^{c_q+1} = 2^{c_a+1}$.

\noindent We first describe a simple terminating protocol that guarantee with high probability $n^{-a} \leq n_e \leq n^a$, for a constant $a$, i.e., the population size estimation is polynomially close to the actual size. Chernoff bounds then imply that repeating this protocol a constant number of times suffices to obtain $n/2 \leq n_e \leq 2n$ with high probability.

\begin{algorithm}[ht]
	\floatname{algorithm}{Protocol} 
	\caption{Population Size Estimation (PSE)}\label{protocol:estimation}
	\begin{algorithmic}[1000]
		
		\State $Q = \{q,\; a,\; l_{(c_q, c_a)}\}$
		\State $\delta:$
		
		\State $ $
		
		\State $(l_{(0,0)},\; q) \rightarrow (l_{(1,0)},\; a)$
		\State $(a,\; q) \rightarrow (a,\; a)$

		\State $(l_{(c_q,c_a)},\; q) \rightarrow (l_{(c_q+1,c_a)},\; q), \; if \; c_q>c_a$
		\State $(l_{(c_q,c_a)},\; a) \rightarrow (l_{(c_q,c_a+1)},\; a), \; if \; c_q>c_a$
		
		\State $(l_{(c_q,c_a)},\; \cdot) \rightarrow (halt,\; \cdot), \; if \; c_q=c_a$
	\end{algorithmic}
\end{algorithm}

\subsection{Analysis}
\label{subsec:counting_with_unique_leader_analysis}

\begin{lemma}
	\label{lemma:size_estimation}
	When half or less of the population has been infected, with high probability $c_q>c_a$. In fact, $c_q - c_a \approx \ln{(n/2)} - \sqrt{\log{n}} > 0$.
\end{lemma}
\begin{proof}
	We divide the process of the epidemic elimination into rounds $i$, where round $i$ means that there exist $i$ infected nodes in the population. Call an interaction a success if an effective rule applies
	and a new $a$ appears on some node. Let the random variable $X$ be the total number of interactions between the leader $l$ and non-infected nodes $q$, the random variable $Y$ be the total number of interactions between $l$ and infected nodes $a$ and the r.v. $I$ be the total number of interactions in the population until all nodes become infected.
	We also define the r.v. $X_i,\; Y_i$ and $I_i$ to be the corresponding numbers in round $i$. Then, it holds that $X = \sum_{i=1}^{n}X_i,\; Y = \sum_{i=1}^{n}Y_i$ and $I = \sum_{i=1}^{n}I_i$. Finally, let the r.v. $X_{ij}$ and $Y_{ij}$ be independent Bernoulli trials such that for $1 \leq j \leq I_i$, $Pr[X_{ij}=1]=p_{Xi}$, $Pr[X_{ij}=0]=1-p_{Xi}$, $Pr[Y_{ij}=1]=p_{Yi}$ and $Pr[Y_{ij}=0]=1-p_{Yi}$. This means that in every interaction in round $i$, the leader, if chosen, interacts with a $q$ node with probability $p_{Xi}$ and with an $a$ node with probability $p_{Yi}$. Then, it holds that $X_i = \sum_{i=1}^{I_i}X_{ij}$ and $Y_i = \sum_{i=1}^{I_i}Y_{ij}$, where $I_i$ is the number of interactions until a success in round $i$.
	
	\begin{equation*}
		p_{Xi} = \frac{2(n-i)}{n(n-1)}, \; p_{Yi} = \frac{2i}{n(n-1)} \text{ and } p_{Ii} = \frac{2i(n-i)}{n(n-1)}
	\end{equation*}
	
	\noindent We also divide the whole process into two phases; the first phase ends when half of the population has been infected, that is $1 \leq i \leq \frac{n}{2}$ and for the second phase it holds that $\frac{n}{2} + 1 \leq i \leq n$. We shall argue that if the counter $c_q$ reaches a value which is a function of $n$, before the second counter $c_a$ reach $c_q$, the leader gives a good estimation. We use $X^a$ and $Y^a$ to indicate the r.v. $X$ and $Y$ during the first phase and $X^b$, $Y^b$ for the second phase. \\
	
	\noindent For $1 \leq i \leq \frac{n}{2}$ (first phase) and by linearity of expectation we have:
	\begin{equation*}
		\begin{split}
			E[X^a] = E[\sum_{i=1}^{n/2}X_i] = E[\sum_{i=1}^{n/2}\sum_{j=1}^{I_i}X_{ij}] = \sum_{i=1}^{n/2}\sum_{j=1}^{I_i}E[X_{ij}]
		\end{split}
	\end{equation*}
	and by Wald's equation, we have that $E[\sum_{i=1}^{I_i}X_{ij}] = E[I_i]E[X_{ij}]$.
	\begin{equation*}
		\begin{split}
			E[X^a] = \sum_{i=1}^{n/2}\frac{n(n-1)}{2i(n-i)}\frac{2(n-i)}{n(n-1)} = \sum_{i=1}^{n/2}\frac{1}{i} = H_{n/2} = \ln{\frac{n}{2}} + a_{n/2} \geq \ln{\frac{n}{2}}
		\end{split}
	\end{equation*}
	where $H_{n/2}$ denotes the $(\frac{n}{2})$th Harmonic number and $0 < a_n < 1$ for all $n \in \mathbb{N}$ (Euler-Mascheroni constant).
	\begin{equation*}
		\begin{split}
			E[Y^a] & = E[\sum_{i=1}^{n/2}Y_i] = E[\sum_{i=1}^{n/2}\sum_{j=1}^{I_i}Y_{ij}] = \sum_{i=1}^{n/2}\sum_{j=1}^{I_i}E[Y_{ij}]
		\end{split}
	\end{equation*}
	and by Wald's equation, we have that $E[\sum_{i=1}^{I_i}Y_{ij}] = E[I_i]E[Y_{ij}]$.
	\begin{equation*}
		\begin{split}
			E[Y^a] = \sum_{i=1}^{n/2}\frac{n(n-1)}{2i(n-i)}\frac{2i}{n(n-1)} = \sum_{i=1}^{n/2}\frac{1}{n-i} = \sum_{i=1}^{n-1}\frac{1}{i} - \sum_{i=1}^{n/2-1}\frac{1}{i} = H_{n-1} - H_{n/2-1} \approx \ln{2} 
		\end{split}
	\end{equation*}
	
	\noindent By Chernoff Bound, the probabilities that the r.v. $X^a$ is less than $(1 - \delta)E(X^a)$ and more than $(1 + \delta)E(X^a)$ are
	\begin{equation*}
		\begin{split}
			Pr[X^a \leq (1 - \delta)E(X^a)] \leq e^{-\frac{\ln{(n/2)}\delta^2}{2}} = \frac{1}{(\frac{n}{2})^{\delta^2/2}}
		\end{split}
	\end{equation*}
	\begin{equation*}
		\begin{split}
			Pr[X^a \geq (1 + \delta)E(X^a)] \leq e^{-\frac{\ln{(n/2)}\delta^2}{3}} = \frac{1}{(\frac{n}{2})^{\delta^2/3}}
		\end{split}
	\end{equation*}
	\noindent that is, $X^a$ does not deviate far from its expectation. The probability that the r.v. $Y^a$ is more than $(1 + \delta)E(Y^a)$, for $\delta=\frac{3\sqrt{\log{n}}}{\ln{2}}$ is
	\begin{equation*}
		\begin{split}
			Pr[Y^a \geq (1 + \delta)E(Y^a)] \leq e^{-\frac{\ln{2}\frac{3\sqrt{\log{n}}}{\ln{2}}}{3}} = \frac{1}{n^{1/2}}
		\end{split}
	\end{equation*}
	
	\noindent Thus, the leader interacts a constant number of times and w.h.p. less than $(1+\delta)E[Y^a]$ times with $a$ nodes during the first phase (half of the population is infected). In addition, it interacts $O(\log{n})$ times with non-infected nodes w.h.p.. In section \ref{sec:all_experiments}, we have tested our results and the Figure \ref{fig:population_size_estimation_counters} confirms this behavior.
	\noindent During the second phase, the infected nodes are more than the non-infected nodes, thus, eventually, the second counter $c_a$ will reach $c_q$ and the leader terminates. By that time, the first counter will already hold a function of $n$ w.h.p. $(c_q - c_a \approx \ln{(n/2)} - \sqrt{\log{n}} > 0)$.
	
	\begin{corollary}
		\label{corollary:population_size_estimation_1}
		\textit{PSE} does not terminate w.h.p. until more than half of the population has been infected.
	\end{corollary}

	\noindent It now suffices to show that the first counter $c_q$ does not continue to rise significantly. During the second phase, where $\frac{n}{2} + 1 \leq i \leq n$, we have
	\begin{equation*}
	\begin{split}
	E[X^b] & = E[\sum_{i=n/2+1}^{n}X_i] = H_n - H_{n/2} \approx \ln{2}
	\end{split}
	\end{equation*}
	
	\noindent By Chernoff Bound, the probability that the r.v. $X^b$ is more than $(1 + \delta)E(X^b)$, for $\delta=\frac{3\log{n}}{\ln{2}}$ is
	\begin{equation*}
	\begin{split}
	Pr[X^b \geq (1 + \delta)E(X^b)] \leq e^{-\frac{\ln{2}\frac{3\log{n}}{\ln{2}}}{3}} = \frac{1}{n}
	\end{split}
	\end{equation*}
	
	$ $
	\qed
\end{proof}

\begin{lemma}
	\label{lemma:size_estimation_2}
	Our \textit{Population Size Estimation} protocol terminates after $\Theta(\log{n})$ parallel time w.h.p..
\end{lemma}

\begin{proof}
	After half of the population has been infected, it holds that $|c_a - c_q| = \Theta(\log{n})$. When this difference reaches zero, the unique leader terminates. We focus only on the effective interactions, which are always interactions between the leader $l$ and $a$ or $q$ nodes. The probability that an interaction is $(l, a)$ is $p_i = i/n > 1/2$, as more than half of the population is infected. Thus, the probability that an interaction is $(l, q)$ is $q_i = 1 - p_i =(n-i)/n < 1/2$. In fact, the probability $p_i$ is constantly decreasing as the epidemic spreads throughout the population. This process may be viewed as a random walk on a line with positions $[0, \infty)$. The particle starts from position $a\log{n}$ and there is an absorbing barrier at $0$. The position of the particle corresponds to the difference $|c_a - c_q|$ of the two counters and it moves towards zero with probability $p_i>1/2$. By the basic properties of random walks, after $\Theta(\log{n})$ steps, the particle will be absorbed at $0$. Thus, the total parallel time to termination is $\Theta(\log{n})$.

	\begin{corollary}
		\label{corollary:population_size_estimation_2}
		When $c_q = c_a$, w.h.p. $2^{c_q+1}$ is an upper bound on $n$.
	\end{corollary}
	$ $
	\qed
\end{proof}

\section{Leader Election with approximate knowledge of $n$}
\label{sec:leaderelection}

The existence of a \textit{unique leader agent} is a key requirement for many population protocols \cite{AAE08} and generally in distributed computing, thus, having a fast protocol that elects a unique leader is of high significance. In this section, we present our \textit{Leader Election} protocol, giving, at first, an abstract description \ref{subsec:leaderelection_abstract}, the algorithm \ref{subsec:leaderelection_protocol} and then, we present the analysis of it \ref{subsec:leaderelection_analysis}. Finally, we have measured the stabilization time of this protocol for different population sizes and the results can be found in section \ref{sec:all_experiments}.

\subsection{Abstract description}
\label{subsec:leaderelection_abstract}
We assume that the nodes know \textit{an upper bound on the population size $n^b$, where $n$ is the number of nodes and $b$ is any big  constant number}. \\
\noindent All nodes store three variables; the round $e$, a random number $r$ and a counter $c$ and they are able to compute random numbers within a predefined range $[1,m]$. We define two types of states; the leaders ($l$) and the followers ($f$). Initially, all nodes are in state $l$, indicating that they are all potential leaders. The protocol operates in rounds and in every round, the leaders compete with each other trying to survive (i.e., do not become followers). The followers just copy the \textit{tuple} $(r, e)$ from the leaders and try to spread it throughout the population. During the first interaction of two $l$ nodes, one of them becomes follower, a random number between $1$ and $m$ is being generated, the leader enters the first round and the follower copies the round $e$ and the random number $r$ from the leader to its local memory. The followers are only being used for information spreading purposes among the potential leaders and they cannot become leaders again. Throughout this paper, $n$ denotes the \textit{population size} and $m$ \textit{the maximum number that nodes can generate}.\\

\noindent \textbf{Information spreading.} It has been shown that the epidemic spreading of information can accelerate the convergence time of a population protocol. In this work, we adopt this notion and we use the followers as the means of competition and communication among the potential leaders. All leaders try to spread their information (i.e., their round and random number) throughout the population, but w.h.p. all of them except one eventually become followers. We say that a node $x$ wins during an interaction if one of the following holds:
\begin{itemize}
	\item Node $x$ is in a bigger round $e$.
	\item If they are both in the same round, node $x$ has bigger random number $r$.
\end{itemize}
One or more leaders $L$ are in the \textit{dominant state} if their tuple $(r_1, e_1)$ wins every other tuple in the population. Then, the tuple $(r_1, e_1)$ is being spread as an epidemic throughout the population, independently of the other leaders' tuples (all leaders or followers with the tuple $(r_1, e_1)$ always win their competitors). We also call leaders $L$ the \textit{dominant leaders}. \\

\noindent \textbf{Transition to next round.} After the first interaction, a leader $l$ enters the first round. We can group all the other nodes that $l$ can interact with into three independent sets.
\begin{itemize}
	\item The first group contains the nodes that are in a bigger round or have a bigger random number, being in the same round as $l$. If the leader $l$ interacts with such a node, it becomes follower.
	
	\item The second group contains the nodes that are in a smaller round or have a smaller random number, being in the same round as $l$. After an interaction with a node in this group, the other node becomes a follower and the leader increases its counter $c$ by one.
	
	\item The third group contains the followers that have the same tuple $(r, e)$ as $l$. After an interaction with a node in this group, $l$ increases its counter $c$ by one.
\end{itemize}

\noindent As long as the leader $l$ survives (i.e., does not become a follower), it increases or resets its counter $c$, according to the transition function $\delta$. When the counter $c$ reaches $b\log{n}$, where $n^b$ is the upper bound on the population size, it resets it and round $r$ is increased by one. The followers can never increase their round or generate random numbers. \\

\noindent \textbf{Stabilization.} The protocol that we present stabilizes, as the whole population will eventually reach in a final configuration of states. To achieve this, when the round of a leader $l$ reaches $\ceil{\frac{2b\log{n}-\log (b\log^2{n})}{\log{m}}}$, $l$ stops increasing its round $r$, unless it interacts with another leader. This rule guarantees the stabilization of our protocol.

\subsection{The protocol}
\label{subsec:leaderelection_protocol}

In this section, we present our \textit{Leader Election} protocol. We use the notation $p_{r,e}$ to indicate that node $p$ has the random number $r$ and is in the round $e$. Also, we say that $(r_1,e_1)>(r_2,e_2)$ if the tuple $(r_1,e_1)$ wins the tuple $(r_2,e_2)$. A tuple $(r_1,e_1)$ wins the tuple $(r_2,e_2)$ if $e_1>e_2$ or if they are in the same round $(e_1=e_2)$, it holds that $r_1>r_2$.

\begin{algorithm}[ht]
	
	\floatname{algorithm}{Protocol}
	\caption{Leader Election}
	\label{protocol:leader_election}
	\begin{algorithmic}[100]
		\State $Q = \{l, f_{r,e}, l_{r,e}\}: r \in [1, m]$
		\State $\delta:$
		\State 
		\State {\#First interaction between two nodes. One of them becomes follower and the other remains leader. The leader generates a random number $r$ and enters the first round $(e=1)$.}
		\State $(l, l) \rightarrow (l_{r,1}, f_{r,1})$ 
		\State
		
		\State {\#A leader in round 0 always loses (i.e., becomes a follower) against a node in a higher round.}
		\State $(f_{r,e}, l) \rightarrow (f_{r,e}, f_{r,e})$
		\State $(l_{r,e}, l) \rightarrow (l_{r,e}, f_{r,e}), \; l_{counter}=l_{counter}+1$ 
		\State

		\State {\#The winning node propagates its tuple. If a leader loses, it becomes follower.}
		\State $(f_{r,i}, f_{s,j}) \rightarrow (f_{k,l}, f_{k,l}), \; \textbf{if}\; (r,i) > (s,j)\; \textbf{then}\; (k,l)=(r,i)\; \textbf{else}\; (k,l)=(s,j) $
		\State $(l_{r,i}, l_{s,j}) \rightarrow (l_{k,l}, f_{k,l}), \; l_{counter}=l_{counter}+1, \; \textbf{if}\; (r,i) \geq (s,j)\; \textbf{then}\; (k,l)=(r,i)\; \textbf{else}\; (k,l)=(s,j) $
		\State $(l_{r,i}, f_{s,j}) \rightarrow (f_{s,j}, f_{s,j}), \; \textbf{if}\; (s,j)>(r,i)$
		\State $(l_{r,i}, f_{s,j}) \rightarrow (l_{r,i}, f_{r,i}), \; l_{counter}=l_{counter}+1, \; \textbf{if}\; (r,i)>(s,j)$
		\State $(l_{r,e}, f_{r,e}) \rightarrow (l_{k,j}, f_{k,j}), \; l_{counter}=l_{counter}+1$
		\State

		\State {\#When a leader increases its counter, the following code is being executed. It checks whether it has reached $c\log{n}$. If yes, it moves to the next round, generates a new random number and checks if it has reached the final round in order to terminate.}
		\State $\textbf{if}\; (l_{counter}=b\log{n}) \; \textbf{then}\{$
		\State \tab Increase round;
		\State \tab Generate a new random number between 1 and m;
		\State \tab Reset counter to zero;
		\State \tab $\textbf{if}\; (Round=\ceil{\frac{2b\log{n}-\log (b\log^2{n})}{\log{m}}}) \;\textbf{Stop increasing the round, unless you interact with a leader;}$
		\State $\}$
		
	\end{algorithmic}
\end{algorithm}

\subsection{Analysis}
\label{subsec:leaderelection_analysis}

The leader election algorithm that we propose, elects a unique leader after $O(\frac{\log^2{n}}{\log{m}})$ parallel time w.h.p.. To achieve this, the algorithm works in stages, called \textit{epochs} throughout this paper and the number of potential leaders decreases exponentially between the epochs.
An epoch $i$ starts when any leader enters the $i$th round $(r=i)$ and ends when any leader enters the $(i+1)$th round $(r=i+1)$. Here we do the exact analysis for $m=\log n$. This can be generalized to any $m$ between a constant and $n$.

\begin{lemma}
	\label{lemma:1}
	During the execution of the protocol, at least one leader will always exist in the population.
\end{lemma}
\begin{proof}
	Assume an epoch $e$, in which only one leader $l_1$ with the tuple $(r_1,e_1)$ exists in the population and the rest of the nodes have become followers. In order for $l_1$ to become follower, there should be a follower with a tuple $(r_2,e_2)$, where $(r_2,e_2)>(r_1,e_1)$. But, while the followers can never increase their epoch or generate a new random number, that would imply that there exists another leader $l_2$ with the tuple $(r_2,e_2)$.
	\qed
\end{proof}

\begin{lemma}
	\label{lemma:2}
	Assume an epoch $e$ and $k$ leaders with the dominant tuple $(r, e)$ in this epoch. The expected parallel time to convergence of their epidemic in epoch $e$ is $\Theta(logn)$.
\end{lemma}
\begin{proof}
	Let the random variable $X$ be the total number of interactions until all nodes have the dominant tuple $(r,e)$. We divide the interactions of the protocol into rounds, where round $i$ means that the epidemic has been spread to $i$ nodes. Initially, $i=k$, that is, the $k$ leaders are already infected by the epidemic, but we study the worst case where $k=1$. Call an interaction a success if the epidemic spreads to a new node. Let also the random variables $X_i, 1 \leq i \leq n-1$, be the number of interactions in the $i$th round. Then, $X = \sum_{i=1}^{n-1}X_i$. The probability $p_i$ of success at any interaction during the $i$th round is: \\
	\begin{equation*}
		p_i = \frac{2i(n-i)}{n(n-1)}
	\end{equation*}
	where $i(n-i)$ are the effective interactions and $\frac{n(n-1)}{2}$ are all the possible interactions. By linearity of expectation we have:
	\begin{equation*}
		\begin{split}
			E[X] & = E[\sum_{i=1}^{n-1}X_i] = \sum_{i=1}^{n-1}E[X_i] =\sum_{i=1}^{n-1} \frac{1}{p_i} = \sum_{i=1}^{n-1} \frac{n(n-1)}{2i(n-i)} \\ 
			& = \frac{n(n-1)}{2} \sum_{i=1}^{n-1} \frac{1}{i(n-i)} \\
			& = \frac{n(n-1)}{2} \sum_{i=1}^{n-1} \frac{1}{n}(\frac{1}{i} + \frac{1}{n-i}) \\
			& = \frac{(n-1)}{2} [\sum_{i=1}^{n-1} \frac{1}{i} + \sum_{i=1}^{n-1} \frac{1}{n-i}] \\
			& = \frac{(n-1)}{2}2H_{n-1} \\
			& = (n-1)[ln(n-1)+  a_{n-1}] = \Theta(n\log n)
		\end{split}
	\end{equation*}
	where $H_n$ denotes the $n$th Harmonic number and  $a_n:=H_n-\log n, (n \in \mathbb{N})$ is a decreasing sequence and $0<a_n<1$ for all $n \in \mathbb{N}$ (\textit{Euler-Mascheroni constant}). It terms of parallel time, it holds that $E[\frac{X}{n}] = \frac{E[X]}{n} = \Theta{(\log{n})}$.
	\qed
\end{proof}

\begin{lemma}
	\label{lemma:4}
	If a counter $c$ of a leader $l$ reaches $b\log{n}$, its epidemic will have already been spread throughout the population w.h.p..
\end{lemma}
\begin{proof}
	Let the r.v. $X$ be the total number of interactions until all nodes have been infected by the dominant tuple. By Lemma \ref{lemma:2}, the expected interactions until the epidemic spreads throughout the whole population is $\mu = (n-1)\ln{(n-1)} + \Theta(1)$. By Chernoff Bound and for $\delta = 1/2$, it holds that
	\begin{equation*}
		\begin{split}
			& Pr[X \leq (1-\delta)\mu] \leq e^{\frac{-\delta^2\mu}{2}} \leq e^{-\frac{(n-1)\ln{(n-1)}}{8}} \leq \left( \frac{1}{n-1} \right)^{(n-1)/8}
		\end{split}
	\end{equation*}
	
	\noindent Thus, the interactions per node under the uniform random scheduler until all nodes become infected are w.h.p. $\frac{(n-1)\ln{(n-1)}}{n} < \frac{n\ln{n}}{n} = \ln{n}$. Thus, after $b\log{n}$ interactions, where $n^b$ is the population size estimation and $b$ a large constant, there are no non-infected nodes w.h.p..
	
	\qed
\end{proof}

\begin{theorem}
	\label{theorem:1}
	After $O(\frac{\log{n}}{\log{m}})$ epochs, there is a unique leader in the population w.h.p..
\end{theorem}
\begin{proof}
	Assume an epoch $e$, in which there are $k$ leaders with the dominant tuple $(r,e)$ and $m$ is the biggest number that the leaders can generate. We shall argue that by the end of the next epoch $e+1$, approximately $\frac{k(m-1)}{m}$ leaders will have become followers and approximately $\frac{k}{m}$ leaders will have a new dominant tuple $(r_2,e_2)$. Whenever the $k$ leaders enter to the next epoch $e+1$, they generate a new random number between $1$ and $m$. Let the random variable $X_e$ be the number of leaders that have randomly generated the biggest number in epoch $e$. We view the possible values of the random choices as $m$ bins and we investigate how many leaders shall go to each bin. Assume the sequence of the random numbers $C_i^e, 1 \leq i \leq k$ that the leaders generate in epoch $e$. Let the random variables $X_i^e$ be independent Bernoulli trials such that, for $1 \leq i \leq k$, $Pr[X_i^e = 1] = p_i$ and $Pr[X_i^e = 0] = 1-p_i$ and $X_e = \sum_{i=1}^{k}X_i^e$. The probability that a leader chooses randomly a number is
	\begin{equation*}
		p_i = \frac{1}{m}
	\end{equation*}
	\noindent Then, the expected number of balls in each bin, thus in the biggest bin also ($X_e$) is 
	\begin{equation*}
		\begin{split}
			& \mu = E(X_e) = E(\sum_{i=1}^{k}X_i^e) = \sum_{i=1}^{k}E(X_i^e) = \sum_{i=1}^{k}p_i = \sum_{i=1}^{k}\frac{1}{m} = \frac{k}{m}
		\end{split}
	\end{equation*}
	
	\noindent Assume now inductively that $X_e \geq a\log^2{n}$, where $a>0$ and $m=\log{n}$. By the Chernoff bound and observing that $k \geq ma\log{n} \Rightarrow \frac{k}{m} \geq a\log{n} \Rightarrow \mu \geq a\log{n}$, we prove that the number of the new dominant leaders will be more than or equal to $\frac{k}{m}(1+\delta)$ with a negligible probability.
	\begin{equation*}
		Pr[X_e \geq (1+\delta)\mu] \leq e^{-\frac{\mu\delta^2}{3}} \leq e^{-\frac{a\log{n}\delta^2}{3}} = n^{-\frac{a\delta^2}{3}} = n^{-\phi}
	\end{equation*}
	
	\noindent For $a \geq \frac{9}{\delta^2}$ it holds that $Pr[X_e \geq (1+\delta)\mu] \leq n^{-3}$. Consequently, if we had $X_e$ leaders in epoch $e$, we now shall have no more than $X_{e+1} \leq (1+\delta)\frac{X_e}{m}$ leaders in epoch $e+1$ with probability $Pr[X_{e+1} \leq (1+\delta)\frac{X_e}{m}] \geq 1 - \frac{1}{n^3}$.
	
	\noindent We can now assume that the expected number of leaders between the epochs can be described by the following recursive function.
	\begin{equation}
		G_e=
		\begin{cases}
			\frac{G_{e-1}}{m}, & i \geq 1 \\
			n, & i=0
		\end{cases}
	\end{equation}
	where $G_e = (1+\delta)X_e$. Then, 
	\begin{equation*}
		G_e = \frac{G_{e-1}}{m} = \frac{G_{e-2}}{m^2} = ... = \frac{n}{m^e}
	\end{equation*}
	The number of the expected epochs until at most $a\log^2{n}$ leaders remain in the population is
	\begin{equation*}
		\begin{split}
			& G_t = a\log^2{n} \Rightarrow \frac{G_{t-1}}{m} = a\log^2{n} \Rightarrow \frac{G_{t-2}}{m^2} = a\log^2{n} \Rightarrow ... \Rightarrow \frac{n}{m^t} = a\log^2{n} \Rightarrow \\
			& m^t = \frac{n}{a\log^2{n}} \Rightarrow log_m(m^t) = log_m(\frac{n}{a\log^2{n}}) \Rightarrow t = \log_m{n} - \log_m(a\log^2{n}) \Rightarrow \\
			& t = \frac{\log{n} - \log{(a\log^2{n})}}{\log{m}} \Rightarrow t = \frac{\log{n} - \log{(a\log^2{n})}}{\log{\log{n}}}
		\end{split}
	\end{equation*}
	
	\noindent Let $E(e)$, be the event that in epoch $e$, there are at most $G_e$ dominant leaders. We consider a success when $(E(e) \given E(1) \cap E(2) \cap \ldots \cap E(e-1))$ occurs until we have at most $\log{n}$ leaders. By taking the union bound, the probability to fail after $t = \frac{\log{n} - \log{(a\log^2{n})}}{\log{\log{n}}}$ epochs is
	given by
	\begin{equation*}
		\begin{split}
			& Pr(\text{fail after } t \text{ epochs}) \leq \sum_{i=0}^{t}Pr[\text{fail in epoch i} \given \text{success until (i-1)th epoch}] \\
			& \leq \sum_{i=0}^{t}\frac{1}{n^{\phi}} = \frac{\frac{\log{n} - \log{(a\log^2{n})}}{\log{\log{n}}}}{n^{\phi}} \leq \frac{1}{n^{\phi - 1}} \leq \frac{1}{n^2}
		\end{split}
	\end{equation*} \\
	$ $
	\begin{corollary}
		\label{corollary:1}
		After $t = \frac{\log{n} - \log{(a\log^2{n})}}{\log{\log{n}}}$ epochs, the remaining leaders are at most $a\log^2{n}$ w.h.p..
	\end{corollary}
	
	\noindent We argue that the number of leaders can be reduced from $a\log^2{n}$ to $a\log{n}$ in one round w.h.p.. The expected value of dominant leaders is now $E[X_{t+1}] = a\log{n}$, thus, by the Chernoff Bound it holds that $Pr[X_{t+1} \geq (1+\delta)\mu] \leq e^{-\frac{a\log{n}\delta^2}{3}}$, and for $a \geq \frac{9}{\delta^2}$, $Pr[X_{t+1} \geq (1+\delta)\mu] \leq n^{-3}$. \\
	
	\noindent Assume w.l.o.g. that $m=a\log{n}$ and according to the previous analysis, there exist $k=a\log{n}$ leaders after $t'=\frac{\log{n} - \log{(a\log^2{n})}}{\log{\log{n}}} + 1$ epochs. The expected value of $X_{t'+1}$ is now $\mu = E[X_{t'+1}] = 1$. Thus, by the Markov Inequality, the probability that the number of the dominant leaders in the next epoch are at least $2$ is
	\begin{equation*}
		\begin{split}
			& P(X_{t'+1} \geq 2) \leq \frac{E[X_{t'+1}]}{2} = \frac{1}{2}
		\end{split}
	\end{equation*}
	\noindent The probability that after $\log_m{n}$ epochs, there is no unique leader in the population is 
	\begin{equation*}
		\begin{split}
			& P[\text{at least } 2 \text{ leaders exist after } \log_m{n} \text{ epochs}] \leq (\frac{1}{2})^{\log_m{n}} = \frac{1}{2^{log_m{n}}}
		\end{split}
	\end{equation*}
	
	\noindent The total number of epochs until there exists a unique leader in the population is w.h.p. $\frac{2\log{n} - \log{(a\log^2{n})}}{\log{m}} + 1 = O(\frac{\log{n}}{\log{m}})$.\\
	\qed
\end{proof}

\begin{theorem}
	\label{theorem:2}
	Our \textit{Leader Election} protocol elects a unique leader in $O(\frac{\log^2{n}}{\log{\log{n}}})$ parallel time w.h.p..
\end{theorem}
\begin{proof}
	There are initially $n$ leaders in the population. During an epoch $e$, by Lemma \ref{lemma:2} the dominant tuple spreads throughout the population in $\Theta(\log{n})$ parallel time, by Lemma \ref{lemma:4}  no (dominant) leader can enter to the next epoch if their epidemic has not been spread throughout the whole population before and by Theorem \ref{theorem:1}, there will exist a unique leader after $O(\frac{\log{n}}{\log{m}})$ epochs w.h.p., thus, for $m = b\log{n}$ the overall parallel time is $O(\frac{\log^2{n}}{\log{\log{n}}})$. Finally, by Lemma \ref{lemma:1}, this unique leader can never become follower and according to the transition function in Protocol \ref{protocol:leader_election}, a follower can never become leader again. \\
	\noindent The rule which says the leaders stop increasing their rounds if $r>=\frac{2b\log{n} - \log{(b\log^2{n})}}{\log{m}}$, unless they interact with another leader, implies that the population stabilizes in $O(\frac{\log^2{n}}{\log\log{n}})$ parallel time w.h.p. and when this happens, there will exist only one leader in the population and eventually, our protocol always elects a unique leader. 
\end{proof}

\begin{remark}
	By adjusting $m$ to be any number between a constant and $n$ and conducting a very similar analysis we may obtain a single leader election protocol whose time and space can be smoothly traded off between $O(\log^2 n)$ to $O(\log n)$ time and $O(\log n)$ to $O(n)$ space.		
\end{remark}

\newcommand{\polylog}{\mathrm{polylog}}
\newcommand{\N}{\mathbb{N}}

\section{Experiments}
\label{sec:all_experiments}

We have also measured the stabilization time of our \textit{Leader Election} and \textit{Population Size Estimation using a unique leader} algorithms for different network sizes. We have executed our protocols $100$ times for each population size $n$, where $n=2^i$ and $i=[3,14]$. Regarding the \textit{Leader Election} algorithm which assumes some knowledge on the population size, the results (Figure. \ref{fig:leader_election}) support our analysis and confirm its logarithmic behavior. In these experiments, the maximum number that the nodes could generate was $m=10$. Finally, all executions elected a unique leader in $a\frac{\log^2{n}}{\log{10}}$ parallel time except one in which two leaders existed by that time (eventually, only one leader remained).

\begin{figure}
	\begin{subfigure}[h]{0.47\linewidth}
		\includegraphics[width=\linewidth]{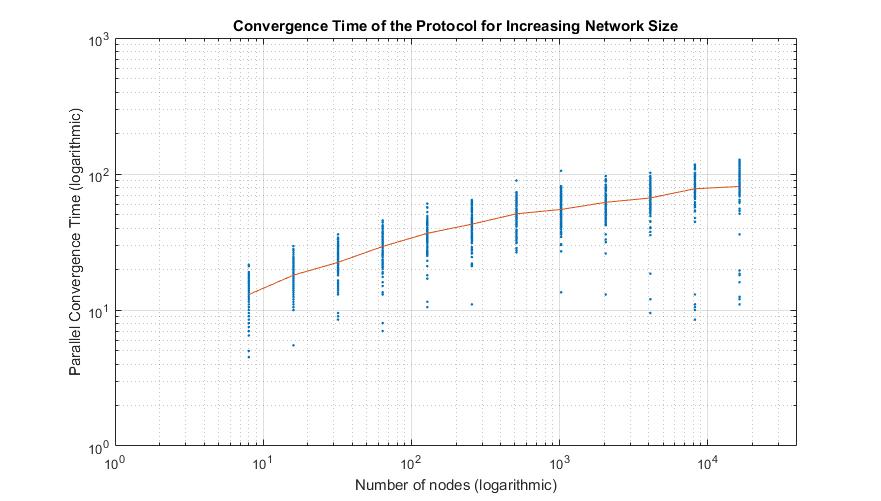}
		\caption{Convergence time.}
		\label{fig:leader_election_time}
	\end{subfigure}
	\hfill
	\begin{subfigure}[h]{0.53\linewidth}
		\includegraphics[width=\linewidth]{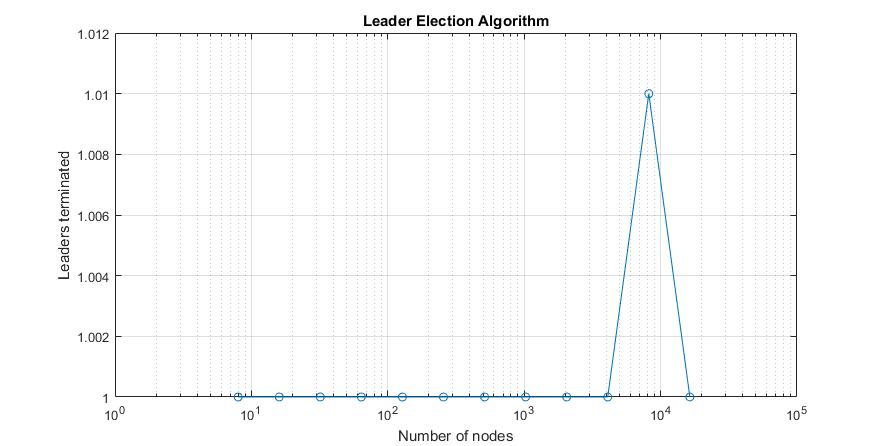}
		\caption{Number of leaders after $\frac{\log^2{n}}{\log{\log{n}}}$ parallel time.}
		\label{fig:leader_election_leaders}
	\end{subfigure}%
	\caption{Leader Election with approximate knowing of $n$. Both axes are logarithmic. In $(a)$ the dots represent the results of individual experiments and the line represents the average values for each network size.}
	\label{fig:leader_election}
\end{figure}

\noindent The stabilization time of our \textit{Approximate Counting with a unique leader} algorithm is shown in Figure \ref{fig:population_size_estimation_time}. The algorithm always gives very close estimations to the actual size of the population (Figure \ref{fig:population_size_estimation_estimations}). Moreover, in Figure \ref{fig:population_size_estimation_counters}, we show the values of the counters $c_q$ and $c_a$, when half of the population has been infected by the epidemic. These experiments support our analysis, while the counter of infected nodes reaches a constant number and the counter of non-infected nodes reaches a value related to $\log{n}$.

\begin{figure}
	\begin{subfigure}[h]{0.5\linewidth}
		\includegraphics[width=\linewidth]{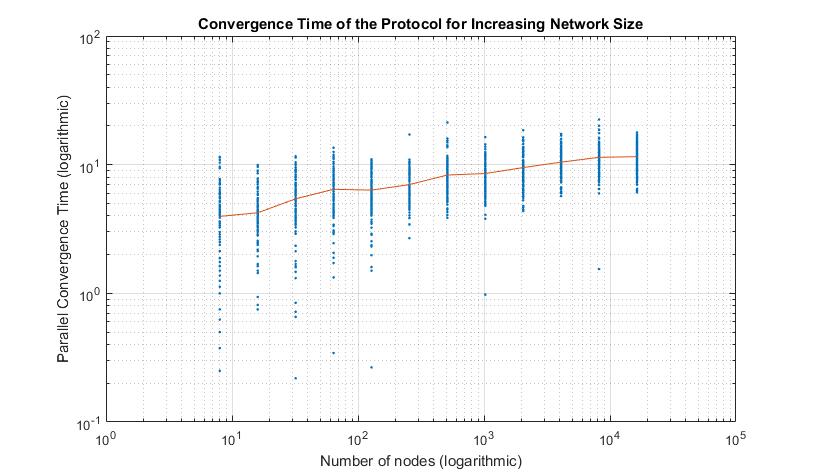}
		\caption{Convergence time.}
		\label{fig:population_size_estimation_time}
	\end{subfigure}
	\hfill
	\begin{subfigure}[h]{0.5\linewidth}
		\includegraphics[width=\linewidth]{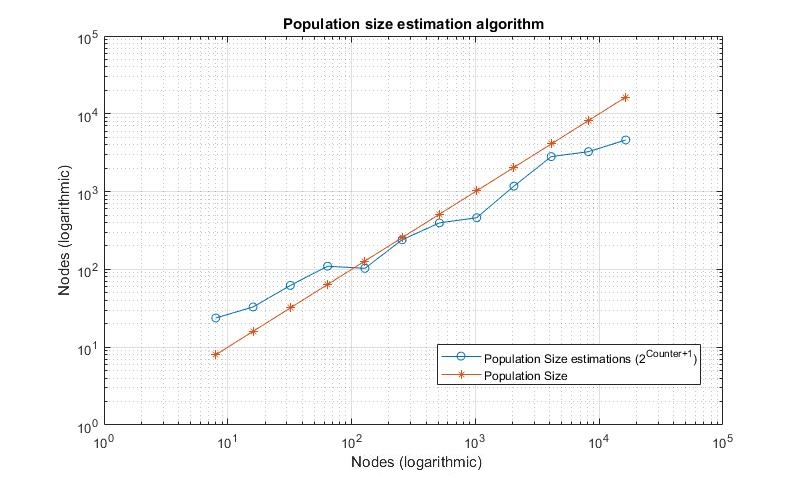}
		\caption{Estimations and actual sizes of the population.}
		\label{fig:population_size_estimation_estimations}
	\end{subfigure}%
	\caption{Approximate Counting with a unique leader. Both axes are logarithmic. In $(a)$ the dots represent the results of individual experiments and the line represents the average values for each network size.}
\end{figure}

\begin{figure}
	\centering
	\includegraphics[width=0.8\textwidth]{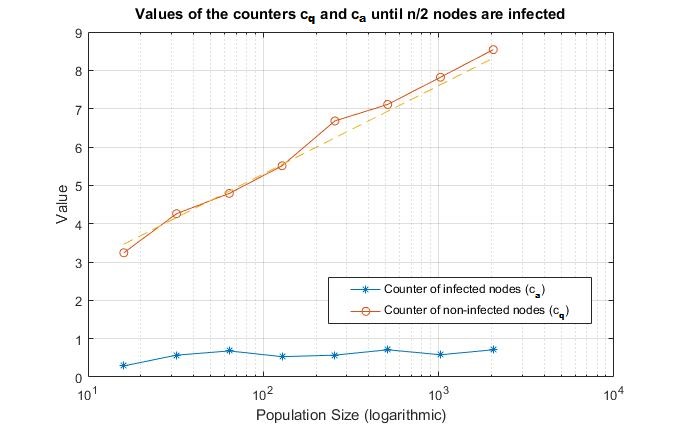}
	\caption{Counters $c_q$ and $c_a$ when half of the population has been infected by the epidemic.}
	\label{fig:population_size_estimation_counters}
\end{figure}

\section{Open Problems}
\label{sec:conclusions}

Call a population protocol \emph{size-oblivious} if its transition function does not depend on the population size. Our leader election protocol requires a rough estimate on the size of the population in order to elect a leader in polylogarithmic time. In addition, our approximate counting protocol requires a unique leader who initiates the epidemic process and then gives an upper bound on the population size. Is it possible to completely drop these assumptions by composing our protocols (i.e., design a size-oblivious and leaderless protocol)?

Moreover, in our leader election protocol, when two nodes interact with each other, the amount of data which is transfered is $O(max\{\log{\log{n}}, \log{m}\})$ bits. In certain applications of population protocols, the processes are not able to transfer arbitrarily large amount of data during an interaction. Can we design a polylogarithmic time population protocol for the problem of leader election that satisfies this requirement?

\vspace{1cm}

\noindent \textbf{Acknowledgments} We would like to thank David Doty and Mahsa Eftekhari for their valuable comments and suggestions during the development of this research work.

\newpage

\bibliographystyle{unsrt}
\bibliography{FullPaper}


\end{document}